\documentclass[aps,prx,reprint,superscriptaddress,twocolumn]{revtex4-2}

\usepackage{graphicx}  

\usepackage{enumerate}
\usepackage[T1]{fontenc}
\usepackage{dcolumn}   
\usepackage{bm}       
\usepackage{amsfonts} 
\usepackage{amsmath}   
\usepackage{amssymb}   
\usepackage{dsfont}
\usepackage{amsthm}
\usepackage{tikz}
\usetikzlibrary{quantikz}
\usepackage{glossaries}
\usepackage{overpic}
\usepackage{color}
\usepackage[colorlinks]{hyperref}
\usepackage{orcidlink}

\definecolor{mycolor}{rgb}{0.122, 0.435, 0.698}
\usepackage{subfigure}
\usepackage{makecell}

\hypersetup{
	bookmarksnumbered,
	pdfstartview={FitH},
	citecolor={darkgreen},
	linkcolor={darkred},
    linktoc={page},
	urlcolor={darkblue},
	pdfpagemode={UseOutlines}}
\definecolor{darkgreen}{RGB}{50,190,50}
\definecolor{darkblue}{RGB}{0,0,190}
\definecolor{darkred}{RGB}{238,0,0}
\definecolor{mycolor}{rgb}{0.122, 0.435, 0.698}
\usepackage{algorithm}
\usepackage{algpseudocode}

\newtheorem{theorem}{Theorem}
\newtheorem{proposition}{Proposition}
\newtheorem{definition}{Definition}
\newtheorem{lemma}{Lemma}
\newtheorem{corollary}{Corollary}

\newcommand{\superUprime}{\hat{\mathcal{U}}^{\prime}}
\newcommand{\superU}{\hat{\mathcal{U}}}
\newcommand{\superD}{\mathcal{D}}
\newcommand{\superLambda}{\Lambda}
\newcommand{\nTG}{n_{\hspace*{1pt}\scriptscriptstyle\mathrm{TG}}} 
\newcommand{\Tr}{\mathrm{Tr}}
\newcommand{\HH}{\mathcal{H}}
\newcommand{\RR}{\mathcal{R}}
\newcommand{\UU}{\mathcal{U}}
\newcommand{\OO}{\mathcal{O}}

\usepackage{soul}

\newcommand{\ins}[1]{\textcolor{blue}{#1}}
\makeatother

\begin{document}

\title{Exploring Noisy Quantum Thermodynamical Processes via the Depolarizing-Channel Approximation}

\author{Jian Li\,\orcidlink{0009-0003-8877-0869}}
\affiliation{Technische Universit{\"a}t Wien, Atominstitut \& Vienna Center for Quantum Science and Technology (VCQ),  Stadionallee 2, 1020 Vienna, Austria}
\affiliation{Department of Physics, University of Vienna, Boltzmanngasse 5, 1090 Vienna, Austria}

\author{Xiaoyang Wang \,\orcidlink{0000-0002-2667-1879}}
\affiliation{RIKEN Center for Interdisciplinary Theoretical and Mathematical Sciences (iTHEMS), Wako 351-0198, Japan}
\affiliation{RIKEN Center for Computational Science (R-CCS), Kobe 650-0047, Japan}

\author{Marcus Huber \,\orcidlink{0000-0003-1985-4623}}

\affiliation{Technische Universit{\"a}t Wien, Atominstitut \& Vienna Center for Quantum Science and Technology (VCQ),  Stadionallee 2, 1020 Vienna, Austria}
\affiliation{Institute for Quantum Optics and Quantum Information - IQOQI Vienna, Austrian Academy of Sciences, Boltzmanngasse 3, 1090 Vienna, Austria}

\author{Nicolai Friis\,\orcidlink{0000-0003-1950-8640}}

\affiliation{Technische Universit{\"a}t Wien, Atominstitut \& Vienna Center for Quantum Science and Technology (VCQ),  Stadionallee 2, 1020 Vienna, Austria}

\author{Pharnam Bakhshinezhad\,\orcidlink{0000-0002-0088-0672}}
\email{pharnam.bakhshinezhad@tuwien.ac.at} 
\affiliation{Technische Universit{\"a}t Wien, Atominstitut \& Vienna Center for Quantum Science and Technology (VCQ),  Stadionallee 2, 1020 Vienna, Austria}

\begin{abstract}  
Noise and errors are unavoidable in any realistic quantum process, including processes designed to reduce noise and errors in the first place. In particular, quantum thermodynamical protocols for cooling can be significantly affected, potentially altering both their performance and efficiency. Analytically characterizing the impact of such errors becomes increasingly challenging as the system size grows, particularly in deep quantum circuits where noise can accumulate in complex ways. To address this, we introduce a general framework for approximating the cumulative effect of gate-dependent noise using a global depolarizing channel. We specify the regime in which this approximation provides a reliable description of the noisy dynamics. Applying our framework to the thermodynamical two-sort algorithmic cooling (TSAC) protocol, we analytically derive its asymptotic cooling limit in the presence of noise. Using the cooling limit, the optimal cooling performance is achieved by a finite number of qubits{\textemdash}distinguished from the conventional noiseless TSAC protocol by an infinite number of qubits{\textemdash}and fundamental bounds on the achievable ground-state population are derived. This approach opens new avenues for exploring noisy quantum thermodynamical processes.
\end{abstract}

\maketitle 


\section{Introduction}

{\noindent}Quantum information processing requires precise control over fragile quantum systems. Tasks such as computation, communication, and sensing often assume ideal operations such as perfect gates, measurements, and timing. In reality, noise from imperfect control and
interactions with the environment inevitably degrade performance. To compensate for noise it is paramount to understand these imperfections, identify error sources, characterize their effect, and develop suitable mitigation strategies. 
Provided that error probabilities are below certain thresholds, quantum error correction can effectively combat the corresponding errors (see, e.g.,~\cite{LidarBrun2013,Terhal2015,Roffe2019}). 
Yet, devices operating in the noisy intermediate-scale quantum (NISQ)~\cite{Preskill:2018} regime studied over the past decade have still been far away from these thresholds, and only recently first devices have been reported to operate at or below threshold~\cite{GoogleQuantumAI2025}.
In addition, many practical computations are performed without the protection of quantum error correction due to severely limited quantum computational resources. 
However, in some cases the effects of noise on the results of quantum computations can be predicted and partially 
ameliorated using quantum error-mitigation schemes~\cite{Kandala2017,TemmeBravyiGambetta2017,EndoBenjaminLi2018,FunckeHartungJansenKuehnStornatiWang2022,CaiBabbushBenjaminEndoHugginsLiMcCleanOBrien2023} that leverage physically inspired assumptions on the noise models. 
Indeed, such mitigation schemes have  already been applied to demonstrate the utility of NISQ devices~\cite{Kim2023,Cao2023,Guo2024}. 

In this work, using similar physically inspired assumptions, we study the effect of general noise on a process that is itself a primitive of most (if not all) of the aforementioned tasks: the preparation of sufficiently pure initial states. In the context of quantum thermodynamics this task is routinely referred to as cooling~\cite{Reeb_2014, Clivaz_2019, Taranto_2023}, as other pure states can be reached from the ground state by suitable unitary transformations. 
Specifically, we focus on  unitary thermodynamical processes implemented via quantum circuits, with finite resources that can be systematically tracked. 
While realistic noise, which takes into account effects such as dephasing, imprecise timekeeping, and gate imperfections is difficult to treat analytically, the depolarizing noise model captures the average effect of more realistic noise models and allows for analytical evaluation~\cite{Urbanek_21,Vovrosh_21,TakahashiTakeuchiTani2020}. 
These averaging effects have been analyzed using random quantum circuits~\cite{Dalzell_24}. However, it is unclear how reasonable this approach is for quantum circuits with fixed gates, which are used in many state-preparation and evolution circuits, as the averaging of realistic noise as well as the corresponding requirements on the quantum circuits have not yet been systematically studied. Here, we therefore introduce and employ a framework based on what we refer to as the global depolarizing approximation (GDA) to describe how realistic noise is averaged in deep quantum circuits with fixed gates. We show that the depolarizing noise well approximates realistic noise as long as the circuit depth grows polynomially with the system size.

Using the GDA, we can efficiently predict the noise effect in deep quantum circuits. As an example, we apply the GDA to the two-sort algorithmic cooling (TSAC) protocol~\cite{Raeisi2019, Taranto_2020}. We derive an analytical expression for its steady-state cooling limit under noise and validate our model by comparing it with numerical simulations of gate-dependent noise and a mirror-cooling protocol~\cite{Oftelie2024}. Our framework effectively predicts final cooling limits, cooling dynamics, and optimal qubit numbers under noise, and extends to other noise types via a generalized depolarizing approximation.\\[-2mm]


\section{The Global Depolarizing Approximation}
\label{sec:gda}

{\noindent}Noise in quantum circuits can be classified into coherent and incoherent noise. Coherent noise due to inaccurate quantum control preserves state purity, while incoherent noise reduces purity mainly coming from randomness in quantum operations and entanglement of qubits with their environment. By applying additional random Pauli gates, coherent noise can be efficiently converted to incoherent noise using randomized compiling~\cite{Wallman2016,Cai_2020}. Thus, we consider incoherent noise throughout this work, which can be modeled as follows
\begin{definition}[Circuit with incoherent noise]
\label{def:noisy_circuit}
A circuit with incoherent noise is composed of $L$ noisy gates and represented by a superoperator $\hat{\mathcal{U}}^{\,\prime} = \bigcirc_{l=1}^L \mathcal{M}_{g_l}$. Each noisy gate $\mathcal{M}_{g_l}$ has a subscript $g_l$ denoting the basic gate type and the qubits it is applied on. Without loss of generality, $\mathcal{M}_{g_l}$ has the decomposition $\mathcal{M}_{g_l} = \mathcal{E}_{g_l} \circ \hat{\mathcal{U}}_{g_l}$~\cite{Magesan2011}, where $\hat{\mathcal{U}}_{g_l}(\rho) = U_{g_l} \rho \,U_{g_l}^{\dagger}$ is the noiseless gate. The error channel $\mathcal{E}_{g_l}$ is given by
\begin{equation} \label{eq:error_channel_form}
    \mathcal{E}_{g_l}(\rho) = (1-p_{g_l}) \rho + p_{g_l} \Lambda_{g_l}(\rho),
\end{equation}
where  $0<p_{g_l}<1$ is the error probability, and $\Lambda_{g_l}(\rho)\neq \rho$ is the noise superoperator of the gate type $g_l$. Additionally, we assume that quantum gates of the same type but applied to different qubits have the same $p_{g_l}$ and the same $\Lambda_{g_l}$.
\end{definition}

The incoherent noise channels $\Lambda_{g_l}$ in a deep circuit can be averaged to a depolarizing channel (Fig.~\ref{fig:illustration}). To this end, we introduce two additional assumptions: \vspace{-2mm}
\begin{enumerate}[(i)]
\item{\label{item i}
Two-qubit gates are the dominant source of error, and only one two-qubit gate type $g$ is used to compose the circuit.  }\\[-3mm]
\item{\label{item ii}
The circuit is sufficiently deep, and its noiseless gate set $\{U_{g_l}\}$ is random enough to form a unitary 2-design, i.e., it resembles a Haar-random circuit ensemble up to the second-order moments~\cite{Dankert2009}.}\\[-4mm]
\end{enumerate}

\vspace{-2mm}
The first assumption (\ref{item i}) holds for most architectures of digital quantum computing~\cite{Arute2019, pino2021, Jurcevic_2021}, since one two-qubit gate type is enough for the universality of the basic gate set~\cite{Nielsen2000}. With this assumption, only one type of gates has non-zero $p_{g_l}$ in Eq.~\eqref{eq:error_channel_form} and one $\Lambda_{g_l}$ is non-trivial. To simplify the notation, the non-zero $p_{g_l}$ and the non-trivial $\Lambda_{g_l}$ are denoted by~$p$ and $\Lambda$, respectively. The second assumption (\ref{item ii}) holds for many dynamic and state-preparation circuits~\cite{Wang:2022OZ,PhysRevResearch.7.023032,yada2025nonhaarrandomcircuitsform}. 
Using these two assumptions, we derive the following theorem:
 
\begin{theorem}[Global depolarizing approximation]
\label{thm:gda}
Under the assumptions~(\ref{item i}) and~(\ref{item ii}), the noisy circuit superoperator $\hat{\mathcal{U}}^{\,\prime}$ (Definition~\ref{def:noisy_circuit}) acting on an initial state $\rho_0$ can be approximated by a global depolarizing channel:
\begin{equation} \label{eq:global_depolarizing_approx}
    \hat{\mathcal{U}}^{\,\prime}(\rho_0) \approx (1 - \eta ) \,\hat{\mathcal{U}}(\rho_0) + \eta \,\tfrac{\mathbb{I}}{d},
\end{equation}
where $\hat{\mathcal{U}}=\bigcirc_{l=1}^L \hat{\mathcal{U}}_{g_l}$ is the superoperator representing the ideal unitary circuit, $d=2^n$ is the Hilbert-space dimension, and~$\eta$ is the effective depolarizing strength given by
\begin{equation}
    \eta := p \, \nTG (1-q).
    \label{eq:depolarizing-probability}
\end{equation}
Here, $\nTG$ is the total number of two-qubit gates in $\hat{U}$, $q:= [\Tr(\Lambda) - 1]/(d^2 -1)$ is a parameter related to the trace of the noise process $\Lambda$ of the two-qubit gates~\cite{Emerson_2005}.
\end{theorem}

\begin{figure}[t!]
    \centering
    \includegraphics[scale=0.75]{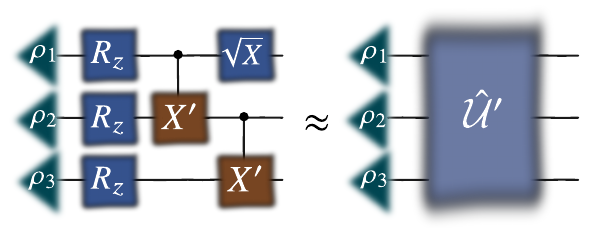}
    \caption{\label{fig:illustration}
    Global Depolarizing Approximation. Local noises of two-qubit CNOT gates are twirled by the quantum circuit with fixed gates (left panel), making them effectively act as a global depolarizing channel $\hat{\UU}'$ on the observable  (right panel).}
\end{figure}

The proof of this theorem is given in Appendix~\ref{app:proof of theorem 1}.
From Theorem~\ref{thm:gda} we see that, when the assumptions~(\ref{item i}) and~(\ref{item ii}) hold, Eq.~\eqref{eq:global_depolarizing_approx} can be the key formula to predict the effect of noise in deep quantum circuits. For example, the depolarizing channel in Eq.~\eqref{eq:global_depolarizing_approx} allows us to establish a relation between the noiseless expectation and the noisy one, i.e.,
\begin{equation}
    \Tr(O\hat{\mathcal{U}}(\rho_0)) = ({1-\eta})^{-1}\,{\Tr(O\hat{\mathcal{U}}^{\,\prime}(\rho_0))}
    \label{eq:noisy-expectation}
\end{equation}
for a Pauli\ins{-}string observable~$O$, and we can directly calculate~$\eta$ based on the number of two-qubit gates $\nTG$ and the noise model $\Lambda$ using Eq.~\eqref{eq:depolarizing-probability}. Previous works use Eq.~\eqref{eq:noisy-expectation} to predict noiseless expectation from the noisy measurement results, conditioned on an additional estimation circuit to effectively evaluate~$\eta$~\cite{Urbanek_21,Vovrosh_21}. Here, we do not rely on the additional estimation circuit, but use Eq.~\eqref{eq:depolarizing-probability} instead that can have less measurement consumption.


\emph{Validity of GDA}{\textemdash}The GDA relies on the circuit's gate set forming an approximate unitary 2-design. On the one hand, the exact 2-design requires a circuit depth that scales exponentially with the system size~\cite{9435351}. On the other hand, if we focus only on estimating physical quantities such as the fidelity and observable expectations with an additive error, it can be shown that the required circuit depth grows polynomially with the system size.

Using Hoeffding's inequality~\cite{Hoeffding01031963}, we prove in Appendix~\ref{app:Achievability of the 2-design approximation},
that to estimate an observable~$O$ within an additive error $\|\langle O\rangle_{\text{meas}}-\langle O\rangle\|\leq \varepsilon$ and a confidence level $1-\delta$, the required circuit depth scales as
\begin{align}
    \text{depth}\sim n(n-1)\frac{2\|O\|}{\varepsilon^2}\ln\left(\frac{2}{\delta}\right), 
\end{align}
where~$n$ is the number of qubits and $\|O\|$ is the spectral norm of the observable~$O$. For fidelity estimation and typical physical observables, where $\|O\|=1$ and $\|O\|$ scales polynomially with~$n$, respectively, the depth grows polynomially in $n$. Thus the GDA is scalable to large system sizes, and is validated in quantum algorithms with a polynomially increased circuit depth, as required by many dynamic and state-preparation quantum algorithms~\cite{RevModPhys.86.153,Daley_2022,RevModPhys.90.015002}. Even if the circuit is not deep enough, one can artificially increase the circuit depth without changing the ideal circuit $\hat{\mathcal{U}}$ using, e.g., unitary gate folding~\cite{TemmeBravyiGambetta2017,hour2024improvingzeronoiseextrapolationquantumgate}. 


\section{Application to Two-Sort Algorithmic Cooling}
\label{sec:application_tsac}

{\noindent}We now demonstrate the utility of the GDA by applying it to analyze the TSAC protocol~\cite{Raeisi2019}, a quantum algorithmic-cooling method to increase the purity of quantum states. We first review the ideal noiseless protocol, then show how our noise model modifies its core dynamics, and finally derive the analytical cooling limit in this realistic setting.\\[-3mm]

The TSAC protocol improves the purity and polarization of $n_{\mathrm{c}}$ computational qubits by introducing an additional reset qubit. The total system therefore consists of $n = n_{\mathrm{c}} + 1$ qubits. The total Hilbert space is $\HH = \HH_{\mathrm{C}}\otimes\HH_{\mathrm{R}}$ representing the $n_{\mathrm{c}}$-qubit computational space $\HH_{\mathrm{C}}$ and the Hilbert space $\HH_{\mathrm{R}}$ of the single reset qubit, respectively. The protocol iteratively applies the following two steps.\\[-2mm]

{\noindent}\textbf{1. Reset step:} The reset qubit is reset to a fresh thermal state $\rho_{\mathrm{R}} := z^{-1} \operatorname{diag}(e^{\epsilon}, e^{-\epsilon})$, 
where $z= e^{\epsilon} + e^{-\epsilon}$ and $\epsilon>0$ is the polarization parameter of the reset qubit.\\[-3mm]

This operation traces out the current state of the reset qubit while preserving the computational subsystem, and then initializes the reset qubit to the fresh thermal state $\rho_R$. Mathematically, this channel $\mathcal{R}$ acting on the full density matrix $\rho$ is defined as $\mathcal{R}(\rho) := \rho_C \otimes \rho_R$, where $\rho_C = \Tr_R(\rho)$ denotes the reduced density matrix of the computational qubits, consistent with the Hilbert space decomposition $\mathcal{H} = \mathcal{H}_C \otimes \mathcal{H}_R$.\\[-2mm]

{\noindent}\textbf{2. Compression step:} The compression unitary $U_{\mathrm{TS}}$ acts on the computational and reset qubits written in the computational basis $\{\ket{0},\ket{1}\}^{\otimes n_{\mathrm{c}}}\otimes \{\ket{0},\ket{1}\}$

\begin{equation}\label{eq:compression_op}
U_{\mathrm{TS}} = 1 \oplus \sigma_x \oplus \cdots \oplus \sigma_x \oplus 1 ,
\end{equation}
where $\sigma_x$ is the Pauli-$X$ matrix and the first and the last elements of the matrix are one. This matrix is a $2^{n}\times 2^{n}$ block-diagonal matrix with $(2^{n_{\mathrm{c}}}-1)$ Pauli-$X$ matrices on its diagonal. The ideal noiseless TSAC protocol iteratively applies $\mathcal{R}$ and $U_{\mathrm{TS}}$ to the density matrix of the~$n$ qubits. 

To study the cooling effect of TSAC, we focus on the density matrix of the $n_{\mathrm{c}}$ computational qubits. The dynamics of the diagonal density-matrix elements can be described as a Markov chain. We denote the $d_{\mathrm{c}}=2^{n_{\mathrm{c}}}$ diagonal elements after $t$ iterations of TSAC as $\vec{v}^{\hspace*{2pt}t}=\{v_1^{\,t}, \dots, v_{d_{\mathrm{c}}}^{\,t} \}$. The transition from $\vec{v}^{\hspace*{2pt}t}$ to $\vec{v}^{\hspace*{2pt}t+1}$ can be described by a transition matrix $T$ given in~\cite{Raeisi2019, Farahmand2022} , such that $\vec{v}^{\hspace*{2pt}t+1} = T\vec{v}^{\hspace*{2pt}t}$. In the cooling limit $t\to \infty$, the steady state $\vec{v}^{\hspace*{2pt}\infty}$ is the eigenvector of~$T$ with eigenvalue~$1$. 

The effect of noise on the TSAC protocol can be studied in the GDA framework. The TSAC protocol can be translated to the digital quantum-circuit model by decomposing $U_{\mathrm{TS}}$ into CNOT and single-qubit gates~\cite{Raeisi2019}, and the resulting $n$-qubit quantum circuit has a layered structure that is sufficiently deep and random to satisfy assumptions (\ref{item i}) and~(\ref{item ii}), see Appendix~\ref{app:derivation_gda} for details.
In particular, the noise after one TSAC iteration can be effectively described by a global depolarizing channel given by Eq.~\eqref{eq:global_depolarizing_approx}, acting on the full $n$-qubit space, with depolarizing strength~$\eta$. This depolarizing channel modifies the matrix $T$ to the noisy $d_{\mathrm{c}}\times d_{\mathrm{c}}$ transition matrix $T^{\,\prime}_{\eta}$ as
\begin{equation}\label{eq:transition_noisy_main}
        T^{\,\prime}_{\eta}=\begin{pmatrix}
        a_{\eta,+} & a_{\eta,+} & \eta/d_{\mathrm{c}} & \cdots & \eta/d_{\mathrm{c}} & \eta/d_{\mathrm{c}} \\[1mm]
        a_{\eta,-} & \eta/d_{\mathrm{c}} & a_{\eta,+} & \cdots & \eta/d_{\mathrm{c}} & \eta/d_{\mathrm{c}} \\[1mm]
        \eta/d_{\mathrm{c}} & a_{\eta,-} & \eta/d_{\mathrm{c}} & \cdots & \eta/d_{\mathrm{c}} & \eta/d_{\mathrm{c}} \\[1mm]
        \vdots & \vdots & \vdots & \ddots & \vdots & \vdots \\[1mm]
        \eta/d_{\mathrm{c}} & \eta/d_{\mathrm{c}} & \eta/d_{\mathrm{c}} & \cdots & \eta/d_{\mathrm{c}} & a_{\eta,+}\\[1mm]
        \eta/d_{\mathrm{c}} & \eta/d_{\mathrm{c}} & \eta/d_{\mathrm{c}} & \cdots & a_{\eta,-} & a_{\eta,-}
        \end{pmatrix},
\end{equation}
where $a_{\eta,\pm} := (\eta/d_{\mathrm{c}})+(1-\eta)e^{\pm\epsilon}/z$. It can be checked that $T^{\,\prime}_{\eta=0}$ is the noiseless transition matrix~$T$. Similar to the noiseless case, the cooling limit of TSAC in the presence of noise is obtained from the eigenvector of $T^{\,\prime}_{\eta}$ with eigenvalue~$1$. Thus we arrive at the following proposition.

\begin{proposition}[Cooling limit of noisy TSAC]
\label{prop:noisy_limit}
The TSAC protocol comprising iterative applications of $\RR$, $U_{\mathrm{TS}}$, and a global depolarizing channel with strength $\eta$ has a unique steady state as the cooling limit.  The steady state is a  probability vector $\vec{v}=\{v_1,\ldots, v_{d_{\mathrm{c}}}\}$ with components\vspace{-2mm}
\begin{equation}\label{eq:limit_noisy}
    v_k = z_1 \lambda_1^{k-1} + z_2 \lambda_2^{k-1} + \frac{1}{d_{\mathrm{c}}}.
\end{equation}
The parameters in this solution are fully determined by the physical properties of the system: $(i)$ The eigenvalues $\lambda_1$ and $\lambda_2$ depend on the noise strength~$\eta$, the polarization $\epsilon$, and the number of computational qubits $n_{\mathrm{c}}$. The eigenvalues are related by the noise-independent constraint $\lambda_1 \lambda_2 = e^{-2\epsilon}$. $(ii)$ The coefficients $z_1$ and $z_2$ are normalization factors determined by $\lambda_1$ and $\lambda_2$, and thus also depend on $(\eta, \epsilon, n_{\mathrm{c}})$.
\end{proposition}

{\noindent}\textit{Sketch of Proof.}{\textemdash}The proof involves solving the steady-state eigenvalue problem $T^{\,\prime}_{\eta}\vec{v} = \vec{v}$. This matrix equation translates into a system of linear equations, which can be formulated as a second-order linear inhomogeneous recurrence relation for the components $v_k'$. The general solution to such a recurrence relation is the sum of a particular solution (a constant related to the uniform background $1/d_{\mathrm{c}}$) and a homogeneous solution, which is a linear combination of terms involving the two roots, $\lambda_1$ and $\lambda_2$, of the characteristic equation. The coefficients of this combination, $z_1$ and $z_2$, are then uniquely determined by the boundary conditions of the system, as detailed in 
Appendix~\ref{app:derivation_noisy_limit}. 
 \qed \\[-2mm] 

\noindent This analytical solution reveals the physical mechanism by which noise limits algorithmic cooling. The steady-state distribution is a superposition of a uniform background and two dynamic modes. The leading-order behavior in the noise strength~$\eta$ is $\lambda_1 = 1 + \eta/\tanh\epsilon +\OO(\eta^2) > 1$ and $\lambda_2=e^{-2\epsilon}(1-\eta/\tanh \epsilon)+\OO(\eta^2) < 1$.\\[-3mm]

In the ideal case ($\eta=0$), $\lambda_1$ becomes exactly~$1$, and its coefficient is $z_1 = -1/d_{\mathrm{c}}$. This causes the constant mode to cancel the uniform background $1/d_{\mathrm{c}}$, leaving only the exponentially decaying term driven by $z_2 > 0$. The result is the ideal cooling limit, where cooling performance improves exponentially with the number of qubits~\cite{Raeisi2019}.\\[-3mm]

However, for any non-zero noise, the presence of the exponentially growing mode ($\lambda_1 > 1$) alters the behavior for a large number of computational qubits~$n_{\mathrm{c}}$. To satisfy the physical constraints of probability positivity and normalization, the amplitudes of both modes, $|z_1|$ and $|z_2|$, must be suppressed as the system size $d_{\mathrm{c}}=2^{n_{\mathrm{c}}}$ increases. This forces the entire distribution to flatten towards the uniform high-temperature state, $v_k \approx 1/d_{\mathrm{c}}$.\\[-3mm]

This behavior is in contrast to the ideal scenario where increasing~$n_{\mathrm{c}}$ always improves cooling performance. Our result shows a trade-off: Noise-induced suppression counteracts the benefits of a larger number of qubits. Thus, Proposition~\ref{prop:noisy_limit} implies the existence of an optimal qubit number for which the best cooling performance is achieved. In the following, we will numerically study this optimal regime.\\[-3mm]


{\noindent}\emph{Optimal cooling performance with physical noise{\textemdash}}To verify our analytical results, we numerically study the TSAC protocol with the depolarizing and realistic physical noise models. The depolarizing model is the theoretical noise model leading to Proposition~\ref{prop:noisy_limit}. For the physical noise model,  we consider the gate-level timekeeping noise which dominates the qubit decoherence~\cite{Xureb2023}.

For the physical timekeeping noise with explicit form given in 
Appendix~\ref{app:specific_noise}, 
we calculate its GDA depolarizing strength using Theorem~\ref{thm:gda}, resulting in the following corollary.

\begin{corollary}[GDA parameter for timekeeping noise]
\label{cor:timekeeping_noise}
Consider the two-qubit timekeeping-noise channel for each CNOT gate with noise strength~$p$ (for the explicit form we refer to 
Appendix~\ref{app:specific_noise}). 
The effective global depolarizing strength~$\eta_{\hspace*{0.75pt}\scriptscriptstyle\mathrm{T}}$ (acting on the full $n$-qubit system) in Theorem~\ref{thm:gda} is\vspace{-3mm}
\begin{equation}
    \eta_{\hspace*{0.75pt}\scriptscriptstyle\mathrm{T}} = p \cdot \nTG \cdot \frac{3d^2}{4(d^2 - 1)},
\end{equation}
where $\nTG$ is the number of CNOT gates in the $n$-qubit quantum circuit and $d=2^n$ is the total dimension.
\end{corollary}
\noindent We use $\eta_{\hspace*{0.75pt}\scriptscriptstyle\mathrm{T}}$ to predict the numerical results of algorithmic cooling with timekeeping noise.\\[-3mm]

As discussed before, Proposition~\ref{prop:noisy_limit} implies an optimal qubit number and a corresponding optimal cooling performance due to the trade-off between cooling power and noise accumulation. The cooling performance can be quantified by the final ground-state population of the target qubit (defined as the first computational qubit). $P_n:=\lim_{\hspace*{0.5pt}t\to \infty}\Tr(\rho_{\hspace*{0.5pt}t}\ket{0}\!\!\bra{0})$ where $\rho_{\hspace*{0.5pt}t}$ is the density matrix of the target qubit after~$t$ TSAC iterations, and we define the optimal cooling performance $P_{\max}:=\max_n P_n$ and the optimal qubit number $n_\text{opt} := \mathrm{argmax}_n P_n$. The values of~$n_\text{opt}$ and~$P_{\max}$ change for different noise strengths. Studying these quantities provides guidance for the practical implementation of the TSAC protocol on noisy quantum hardware with some specific noise strength.\\[-3mm]

To determine~$n_\text{opt}$ and~$P_{\max}$, we simulate the TSAC protocol for various noise strengths and total number of qubits. The dependence of the optimal performance metrics~$n_\text{opt}$ and~$P_{\max}$ versus error probability are summarized in Fig.~\ref{fig:optimal_summary}. We assume that the initial state $\rho_i$ of the target qubit has the ground-state population $P_{\text{initial}}=\Tr(\rho_i\ket{0}\!\!\bra{0})=0.85$. Figure~\ref{fig:optimal_summary}~(a) displays the optimal qubit number $n_\text{opt}$ as a function of the CNOT error probability~$p$. The value of $p$ serves as the fundamental noise parameter for both curves. The ``physical'' curve (circles) shows the result of a direct gate-level simulation using the timekeeping-noise model (see
Appendix~\ref{app:specific_noise}) 
with strength~$p$. The ``theoretical'' curve (squares) shows the analytical prediction from Proposition~\ref{prop:noisy_limit}, which uses the GDA strength $\eta = \eta_{\hspace*{0.75pt}\scriptscriptstyle\mathrm{T}}$ calculated from~$p$ via Corollary~\ref{cor:timekeeping_noise}. Figure~\ref{fig:optimal_summary}~(b) shows the corresponding maximum achievable ground\ins{-}state population $P_\text{max}$ with the corresponding $n_\text{opt}$ from Fig.~\ref{fig:optimal_summary}~(a). Results comparing the physical noise simulation and the GDA model show agreement, validating our GDA framework's prediction on algorithmic cooling. Additionally, we see that both~$P_\text{max}$ and~$n_\text{opt}$ decrease monotonically as the error probability increases, which is consistent with the intuition that noise causes the cooling performance to deviate from the ideal noiseless case. Specifically, currently available superconducting chips can achieve $p\sim 10^{-3}$~\cite{PRXQuantum.6.010349,Sete2024}, with the optimal qubit number~$n_\text{opt}=3$. Therefore, unlike in the noiseless case, under the current noise level, further increasing the total number~$n$ of qubits does not improve the cooling performance. 

\begin{figure}[t!]
    \centering
    \includegraphics[scale=1]{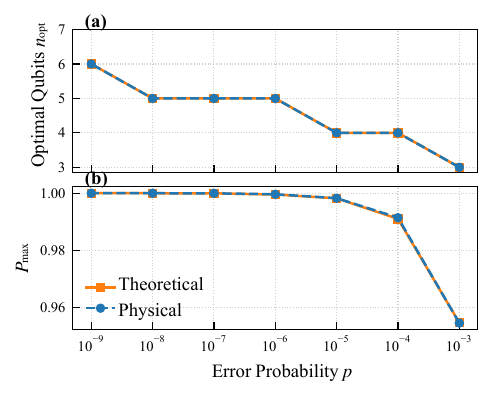}\vspace{-2mm}
    \caption{\label{fig:optimal_summary}
    Summary of optimal performance versus CNOT error probability $p$ (logarithmic scale). 
    \textbf{(a)} Optimal number of total qubits ($n_\text{opt}$) required to maximize the final ground-state population. 
    \textbf{(b)} Corresponding maximally achievable final ground-state population ($P_\text{max}$). 
    All panels compare results from physical noise simulations (circles, dashed lines) and the theoretical GDA model (squares, solid lines).
    The initial ground-state population is $P_{\text{initial}}=0.85$.}
\end{figure}

\section{Application to Dynamic Cooling}
\label{sec:application_dc}

{\noindent}To further analyze the performance of the GDA framework, 
we consider a different class of algorithms known as single-shot or dynamic cooling (DC)~\cite{Schulman1999, Oftelie2024,lin2024thermodynamic,rodriguez2020novel}. Unlike the iterative TSAC protocol using one auxiliary reset qubit,
DC protocols apply a global unitary transformation only once across the target qubit and a set of auxiliary qubits to transfer entropy from the target to the auxiliaries.\\[-3mm]

Specifically, the DC process begins with an initial state given by the tensor product of~$n$ identical thermal states, $\rho_{\hspace*{0.5pt}\text{th}}$, that is
    $\rho_{\hspace*{0.5pt}\text{in}}^{\phantom{\otimes n}} = \rho_{\hspace*{0.5pt}\text{th}}^{\otimes n}$.
A global unitary operator, $U_{\scriptscriptstyle\text{DC}}$, is then applied to the entire system. The final, cooled state of the target qubit, $\rho_{\hspace*{0.5pt}\scriptscriptstyle\mathrm{T}}$, is obtained by tracing out the $n-1$ auxiliary qubits (labeled here as subsystems $2$,$3$, \ldots, $n$):
\begin{equation}
    \rho_{\hspace*{0.5pt}\scriptscriptstyle\mathrm{T}} = \Tr_{2,3,\ldots,n} \left( U_{\scriptscriptstyle\text{DC}}^{\phantom{\dagger}} 
    \,\rho_{\hspace*{0.5pt}\text{in}} \,U_{\scriptscriptstyle\text{DC}}^{\dagger} \right).
\end{equation}
Several specific unitaries $U_{\scriptscriptstyle\text{DC}}$ have been proposed~\cite{Oftelie2024}. For our analysis, we use the computationally efficient mirror protocol~\cite{Oftelie2024}.\\[-3mm]


{\noindent}\emph{Numerical Verification for DC}{\textemdash}We verify the GDA model against a physical noise simulation for the DC mirror protocol. To better resolve performance differences in the high-purity regime and to maintain consistency with related work~\cite{Oftelie2024}, we use the final effective temperature of the target qubit as the performance metric instead of its ground-state population. This effective temperature, $\mathcal{T}_{\text{eff}}$, is derived from the final diagonal elements ($p_0, p_1$) of the target qubit's density matrix $\rho_{\hspace*{0.5pt}\scriptscriptstyle\mathrm{T}}$ in the energy basis (the eigenbasis of $H = \tfrac{1}{2}\hbar\omega\sigma_z$~\cite{Oftelie2024}) via the Boltzmann relation $p_1/p_0 = e^{-\hbar\omega/(k_{\hspace*{0.5pt}\scriptscriptstyle\mathrm{B}} \mathcal{T}_{\text{eff}})}$, where $k_{\hspace*{0.5pt}\scriptscriptstyle\mathrm{B}}$ is the Boltzmann constant.\\[-3mm]

\begin{figure}[t]
    \centering
    \includegraphics[width=\linewidth]{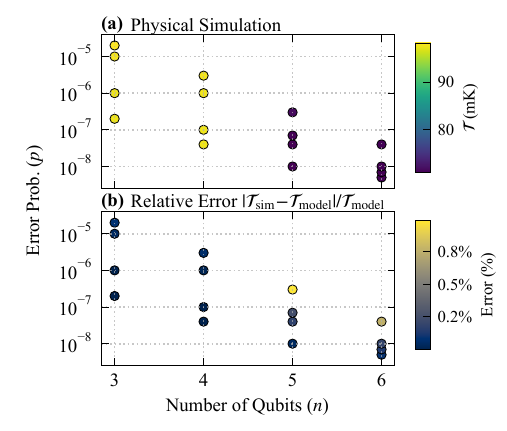}\vspace{-2mm} 
    \caption{Verification of the GDA model for the DC mirror protocol across a range of qubit numbers~$n$ and CNOT error probabilities~$p$. The panels show: (a) the final effective temperature $\mathcal{T}_{\text{sim}}$ obtained from a physical noise simulation; (b) The relative error, $|\mathcal{T}_{\text{sim}}-\mathcal{T}_{\text{model}}|/\mathcal{T}_{\text{model}}$, where $\mathcal{T}_{\text{model}}$ is the result from GDA. Temperatures in (a) are shown in mK; (b) is shown as a percentage. The plots demonstrate agreement of the model with the simulation.}
    \label{fig:dc_mirror_comparison} 
\end{figure}

For these simulations, the qubits are initialized thermal states $\rho_{\hspace*{0.5pt}\text{th}} = e^{-\beta H} / \mathcal{Z}$, where $\mathcal{Z} = \Tr(e^{-\beta H})$ and $\beta = 1/(k_{\hspace*{0.5pt}\scriptscriptstyle\mathrm{B}} \mathcal{T}_{\text{initial}})$~\cite{Oftelie2024}. In the numerical simulation, we use $\omega = 10 \, \mathrm{GHz}$ and an initial temperature $\mathcal{T}_{\text{initial}}=163$~mK ($P_{\text{initial}} \approx 0.95$), which is consistent with typical parameters for superconducting
qubit quantum computers~\cite{Oftelie2024}.

Figure~\ref{fig:dc_mirror_comparison} shows this final effective temperature (in mK) as a function of the qubit number~$n$ and the CNOT error probability~$p$. Figures~\ref{fig:dc_mirror_comparison}~(a) shows the results obtained from the physical noise simulation ($\mathcal{T}_{\text{sim}}$).
The final effective temperature predicted by GDA is denoted as $\mathcal{T}_{\text{model}}$.
We observe that the relative error of the final temperature, $|\mathcal{T}_{\text{sim}}-\mathcal{T}_{\text{model}}|/\mathcal{T}_{\text{model}}$, shown in Fig.~\ref{fig:dc_mirror_comparison}, is at the level of 1\%. This result validates our GDA framework's applicability to the DC mirror protocol. \\[-2mm]


\section{Conclusion and Outlook}
\label{sec:conclusion}

{\noindent}We introduced the GDA, a scalable theoretical framework for analyzing the cumulative effects of generic physical noise in deep quantum circuits. The GDA approximates complex local noise channels with a single, system-wide depolarizing channel, requiring a circuit depth that only scales polynomially with the system size to emulate a unitary 2-design. Applying the GDA to the TSAC protocol, we derived a novel analytical expression (Proposition~\ref{prop:noisy_limit}) for the asymptotic cooling limit under realistic noise. This analysis identified a fundamental noise-induced trade-off: The exponential gain in ideal cooling power with increasing qubit number~$n$ is limited by the exponential increase in noise accumulation in deep quantum circuits. This competition dictates an optimal qubit number of the cooling protocol and sets a sharp upper bound on the achievable ground-state population (Fig.~\ref{fig:optimal_summary}). The accuracy of the GDA in predicting this optimal boundary was also validated against detailed physical noise simulations for both TSAC and DC protocols.

A key direction for future research is the broad application of the GDA framework to other noisy input-output circuit-based protocols, which include most quantum thermodynamical processes. This approach offers a powerful avenue for the analytical investigation of fundamental performance bounds~\cite{Reeb_2014, Lipka_Bartosik_2024, Xuereb_2025} for thermodynamic tasks under realistic noise. Crucially, the GDA offers a critical tool for quantifying the thermodynamic cost of imperfect control. Physical constraints inherent in finite-resource implementations directly lead to operational errors using the depolarizing strength~$\eta$ featuring in the GDA. By analytically linking the thermodynamic resources (e.g., energy, time, or control precision)~\cite{silva2024optimalunitarytrajectoriescommuting, Taranto_2023, Xuereb_2023} required to achieve a target error tolerance to~$\eta$, this framework facilitates the rigorous study of fundamental resource-error-performance trade-offs{\textemdash}a necessary step toward optimizing finite-resource quantum heat engines.


\begin{acknowledgments}
{\noindent}The authors thank Nayeli Rodr\'{\i}guez-Briones, Jake Xuereb, and Philip Taranto for insightful discussions. 
We acknowledge funding from the European Research Council (Consolidator grant ‘Cocoquest’ 101043705), from the Austrian Science Fund (FWF) through the stand-alone project P 36478-N funded by the European Union{\textemdash}NextGenerationEU, as well as by the Austrian Federal Ministry of Education, Science and Research via the Austrian Research Promotion Agency (FFG) through the flagship project FO999897481 (HPQC), the project FO999914030 (MUSIQ), the project FO999921407 (HDcode), and the project FO999921415 (Vanessa-QC) funded by the European Union{\textemdash}NextGenerationEU.  X.W. is supported by the RIKEN TRIP initiative (RIKEN Quantum) and the University of Tokyo Quantum Initiative. 
\end{acknowledgments}


\bibliographystyle{apsrev4-1fixed_with_article_titles_full_names_new}
\bibliography{bibliography.bib}

\clearpage
\onecolumngrid
\appendix

\newpage
\onecolumngrid
\appendix

\section*{Appendices}
\setcounter{section}{0}
\tableofcontents

\setcounter{equation}{0}
\setcounter{figure}{0}
\setcounter{table}{0}

\setcounter{section}{0}

\renewcommand{\theequation}{A\arabic{equation}}
\renewcommand{\thefigure}{A\arabic{figure}}
\renewcommand{\thetable}{A\arabic{table}}

\section{Derivation of the Global Depolarizing Approximation}
\label{app:derivation_gda}
{\noindent}In this section, we provide the detailed derivation for the global depolarizing approximation (GDA). A key assumption in the proof of the GDA is the randomness of the ideal gate set. In the next subsection, we prove that this randomness can be achieved using moderate-depth quantum circuits as long as our aim is to estimate the expectation value of some observables or the state fidelity.


\subsection{Proof of the GDA}\label{app:proof of theorem 1}
{\noindent}The theorem of the GDA is as follows:\\

\noindent\textbf{Theorem 1.}
\label{thm:gda_restated} 
\textit{
Under the assumptions that 
(i) noise acts identically on all two-qubit gates with error probability $p$, and 
(ii) the ideal unitary circuit forms an approximate 2-design, 
the noisy circuit superoperator $\hat{\mathcal{U}}^{\,\prime}$, defined explicitly as
\begin{equation}
    \hat{\mathcal{U}}^{\,\prime} = \bigcirc_{l=1}^L \left[ (1-p_{g_l})\hat{u}_l + p_{g_l}(\Lambda_{g_l} \circ \hat{u}_l) \right],
\end{equation}
acting on an initial state $\rho_0$, can be approximated by a global depolarizing channel:
\begin{equation} \label{eq:global_depolarizing_approx_sm}
    \hat{\mathcal{U}}^{\,\prime}(\rho_0) \approx (1 - \eta ) \,\hat{\mathcal{U}}(\rho_0) + \eta \,\tfrac{\mathbb{I}}{d},
\end{equation}
where $\hat{\mathcal{U}}=\bigcirc_{l=1}^L \hat{\mathcal{U}}_{g_l}$ is the superoperator representing the ideal unitary circuit, $d=2^n$ is the Hilbert-space dimension, and $\eta$ is the effective depolarizing strength given by
\begin{equation}
    \eta := p \, \nTG (1-q).
    \label{eq:depolarizing-probability_sm}
\end{equation}
Here, $\nTG$ is the total number of two-qubit gates in the circuit, and $q:= [\Tr(\Lambda) - 1]/(d^2 -1)$ is a parameter related to the trace of the noise process $\Lambda$ of the two-qubit gates~\cite{Emerson_2005}.
}
\begin{proof}
{\noindent}The starting point is the full expression for the noisy circuit superoperator:
\begin{equation}\label{app:eq:noisy_circuit_expanded}
\hat{\mathcal{U}}^{\prime}=\bigcirc_{l=1}^{L} [(1-p_{g_{l}})\hat{u}_{l}+p_{g_{l}}(\Lambda_{g_{l}}\circ\hat{u}_{l})], 
\end{equation}
where $L$ is the total number of quantum gates. Using our assumptions above that only two-qubit gates have the same noise strength, we set the error probability $p_{g_l}=p$ if the gate at location $l$ is a two-qubit gate, and $p_{g_l}=0$ otherwise. We define $\nTG$ as the total number of such noisy two-qubit gates. Expanding Eq.~\eqref{app:eq:noisy_circuit_expanded} using the binomial theorem gives:
\begin{equation}\label{app:ExpandNoisyProduct}
    \superUprime = P(0)\,\superU + \sum_{M=1}^{\nTG} P(M)\, \superU \,\frac{1}{\binom{\nTG}{M}} \sum_{\substack{1 \le j_0 < j_1 < \dots \\ < j_{M-1} \leq L}} \left( \bigcirc_{\alpha=0}^{M-1} \mathcal{L}_{j_\alpha} \right),
\end{equation}
where $P(M) = \binom{\nTG}{M}p^M(1-p)^{\nTG-M}$ is the binomial probability. The term $\mathcal{L}_{j_\alpha} := \superD_{j_{\alpha}}^{\dagger}\superLambda_{j_{\alpha}}\superD_{j_{\alpha}}$ represents the effective error channel: $\superLambda_{j_{\alpha}}:=\Lambda_{g_{j_\alpha}}$ is the ``bare'' noise channel at the $j_\alpha$-th gate site where the two-qubit gate applies, and $\superD_{j_{\alpha}}$ is the composite noiseless evolution defined by $\superD_{j_\alpha} := \hat{u}_{j_\alpha} \circ \hat{u}_{j_\alpha-1} \circ \dots \circ \hat{u}_1$.

With a small error probability~$p$, the leading-order binomial series in Eq.~\eqref{app:ExpandNoisyProduct} reads 
\begin{equation}\label{app:eq:first_order}
\begin{split}
\hat{\mathcal{U}}^{\,\prime}(\rho_0) = (1-p\,\nTG) \,\hat{\mathcal{U}}(\rho_0) + p\,\nTG\, \hat{\mathcal{U}} \mathcal{S}(\rho_0) + O(p^{2} \nTG^{2}),
\end{split}
\end{equation}
with higher order terms being truncated as long as $p \cdot \nTG \ll 1$. To simplify the notation for the average error operator, at this step we re-index the sequence of these $n_{TG}$ noisy gates as $j = 1, \dots, n_{TG}$. The leading order term is abbreviated by defining the average error operator:
\begin{equation}
\mathcal{S} := \frac{1}{\nTG}\sum_{j=1}^{\nTG} \mathcal{L}_j = \frac{1}{\nTG}\sum_{j=1}^{\nTG}\mathcal{D}_j^\dagger \Lambda_j \mathcal{D}_j. 
\end{equation}
Under our assumption, we focus solely on the noise affecting two-qubit gates and universal gate sets with only one type of two-qubit gate. It follows that all $\Lambda_j$s represent the same noise channel, differing only in the specific pair of qubits upon which they act. Since the number of qubits~$n$ is finite, the possible number of error channels $\Lambda_i$ is finite.  That is, $\Lambda_j \in \{\Lambda^{(1)}\dots\Lambda^{(K)}\}$ where $K$ is the total number of possible two-qubit gates of a digital quantum system and $k \in \{1, \dots, K\}$ indexes the distinct qubit pairs in the system. For example, a digital system with the CNOT gate as the only two-qubit basic gate has $K=n(n-1)$. Suppose that  $\Lambda^{(j)} \neq \mathbb{I}$ appears $m_{j}$ times in the summation $\mathcal{S}$ , which satisfies $\sum_{k=1}^{K} m_k = \nTG$. Then, $\mathcal{S}$ can be re-written as: 

\begin{equation}\label{summation}
    \mathcal{S} = \sum_{k=1}^{K}p_{k}\cdot\frac{1}{m_{k}}\sum_{j=1}^{m_{k}}\mathcal{D}_{j}^{(k)\dagger}\Lambda^{(k)}\mathcal{D}_{j}^{(k)},
\end{equation}
where $p_k := m_k/\nTG$ is normalized as $\sum_{k=1}^{K} p_k=1$.

This average error operator is reduced to the global depolarizing channel by adopting the assumption (ii). With a sufficiently deep quantum circuit, the associated set of operators $\{\mathcal{D}_j^{(k)}\}_{j=1}^{m_k}$ in Eq.~\eqref{summation} forms an approximate 2-design such that~\cite{Dankert2009} 
\begin{equation}\label{eq:2-design}
\frac{1}{m_k}\sum_{j=1}^{m_{k}}\mathcal{D}_{j}^{(k)\dagger}\Lambda^{(k)}\mathcal{D}_{j}^{(k)}(\rho) = q_k\,\rho + (1-q_k)\frac{\mathbb{I}}{d},
\end{equation}
where $q$ is a parameter related to the trace of the noise process superoperator $\Lambda^{(k)}$~\cite{Emerson_2005}: 
\begin{equation} \label{eq:q_definition_2}
    q_k := \frac{\Tr(\Lambda^{(k)}) - 1}{d^2 -1},
\end{equation}
where $d=2^n$ is the dimension of the total Hilbert space of the quantum circuit. Given that the two-qubit gates are of the same type, $\Lambda^{(k)}$ on different qubits has the same trace, i.e., $\Lambda= \Lambda^{(k)}, \forall k \in [1,K]$, such that all $q_k$s have the same value of $q = q_k, \forall k \in [1,K]$. Thus we can simplify Eq.~(\ref{summation}) as:
\begin{equation}\label{eq:summation_final}
 \frac{1}{\nTG}\sum_{j=1}^{\nTG}\mathcal{L}_j(\rho) = q\rho + (1-q)\frac{\mathbb{I}}{d}.
\end{equation}
Combining this equation with Eq.~(\ref{app:eq:first_order}) gives
\begin{equation}
\hat{\mathcal{U}}^{\,\prime}(\rho_0) = \left[1 - p \, \nTG(1-q) \right] \hat{U}(\rho_0)+ p \, \nTG(1-q) \frac{\mathbb{I}}{d}+ O(p^{2}\nTG^{2}).
\end{equation}
This completes our proof of Theorem~\ref{thm:gda} in the main text.
\end{proof}

\subsection{Achievability of the 2-design approximation}\label{app:Achievability of the 2-design approximation}

{\noindent}During the proof of Theorem~\ref{thm:gda}, we require that the sets of operators $\{\mathcal{D}_j^{(k)}\}_{j=1}^{m_k}$ for all $k\in[1,K]$ are random enough, such that they form approximate unitary $2$-design. It has been shown that the non-Haar random circuits form unitary designs as fast as Haar random circuits~\cite{yada2025nonhaarrandomcircuitsform}. Thus, $\{\mathcal{D}_j^{(k)}\}$ are not required to be Haar random. However, even for Haar random circuits, such as Clifford circuits, their cardinality of $n$-qubit systems scales as $\mathcal{O}(4^{n^2})$~\cite{9435351}. Note that the number of elements $m_k$ in $\{\mathcal{D}_j^{(k)}\}$ is proportional to the circuit depth. Thus, one cannot efficiently achieve the unitary $2$-design by simply constructing quantum circuits with depth of $\mathcal{O}(4^{n^2})$. However, we prove in the following content that, if our aim is to estimate the expectation values of some observables or the state fidelity, the unitary $2$-design can be achieved using moderate circuit depth.

Assume that the GDA is used to estimate the expectation
\begin{align}
    \langle O\rangle_{m_k} := \Tr [\frac{1}{m_k}\sum_{j=1}^{m_{k}}\mathcal{D}_{j}^{(k)\dagger}\Lambda^{(k)}\mathcal{D}_{j}^{(k)}(\rho) O],
    \label{eq:O-m-k-definition}
\end{align}
where~$O$ is an Hermitian observable or a state projector $\ket{\psi}\bra{\psi}$ for fidelity estimation. Denote the exact 2-design results as $\langle O\rangle = \lim_{m_k\to\infty} \langle O\rangle_{m_k}$. The probability of deriving $\langle O\rangle_{m_k}$ outside the $\varepsilon$-neighbor of $\langle O\rangle$ is bounded by Hoeffding's inequality
\begin{align}
    \mathrm{Prob}[|\langle O\rangle_{m_k}-\langle O\rangle|\geq \varepsilon]\leq 2e^{-2m_k \varepsilon^2/(b-a)^2},
\end{align}
where $[a,b]$ is the range of each $\Tr [\mathcal{D}_{j}^{(k)\dagger}\Lambda^{(k)}\mathcal{D}_{j}^{(k)}(\rho) O]$ in Eq.~\eqref{eq:O-m-k-definition}, and is bounded by the spectral norm
\begin{align}
    |b-a|\leq 2\,\| O\|.
\end{align}
For example, the state projector has $\|\ket{\psi}\!\!\bra{\psi}\|=1$, and $\|O\|$ is typically $\mathcal{O}(\mathrm{poly}(n))$ for most observables in quantum many-body systems.

Suppose that we want 
\begin{align}
    \mathrm{Prob}[|\langle O\rangle_{m_k}-\langle O\rangle|\geq \varepsilon]\leq \delta,
\end{align}
where $\varepsilon$ denotes the accuracy of the estimate and $1-\delta$ represents the desired confidence level, where~$\delta$ is some constant, e.g., $\delta = 0.05$. Then, the required number of random samples $m_k$ can be obtained by setting $\delta = 2e^{-2m_k \varepsilon^2/(b-a)^2}$, such that
\begin{align}
    m_k = \frac{(b-a)^2}{2\varepsilon^2}\ln\left(\frac{2}{\delta}\right)\leq \frac{2\|O\|}{\varepsilon^2}\ln\left(\frac{2}{\delta}\right).
    \label{eq:mk-estimation}
\end{align}
This $m_k$ bound is explicitly independent of~$n$ for fidelity estimation, and increases at most polynomially to~$n$ for observable expectations. 

If a digital quantum chip has $\text{CNOT}$ gate as the only two-qubit basic gates, since $k$ takes $n(n-1)$ values, we can estimate the GDA circuit depth to achieve the approximate $2$-design:
\begin{align}
    \text{depth}\sim n(n-1)\frac{2\|O\|}{\varepsilon^2}\ln\left(\frac{2}{\delta}\right).
\end{align}
Thus, the unitary $2$-design can be achieved using polynomially increased circuit depth to evaluate fidelity and observable expectations. 


\subsection{Numerical Validation of the 2-design Approximation for TSAC}
\label{app:validation_tsac}

{\noindent}The GDA framework relies on the assumption that the set of prefix unitaries $\{\superD_j^{(k)}\}$ forms an approximate 2-design, which validates the channel-averaging approximation in Eq.~\eqref{eq:2-design}. To verify the effectiveness of this assumption for the specific $U_{\mathrm{TS}}$ unitary used in the TSAC protocol, we perform a direct numerical test.

We quantify the accuracy of the approximation in Eq.~\eqref{eq:2-design} by calculating the fidelity between its left-hand side (LHS) and right-hand side (RHS). We define the averaged channel output for a CNOT location $k$ as $\rho_{\mathrm{LHS}\_k} = \frac{1}{m_{k}}\sum_{j=1}^{m_{k}}\superD_{j}^{(k)\dagger}\Lambda^{(k)}\superD_{j}^{(k)}(\rho_0)$ and the GDA model's approximation as $\rho_{RHS} = q_{k}\rho_0+(1-q_{k})\frac{\mathbb{I}}{d}$. We then compute the average fidelity over all $K$ CNOT locations, $\langle F_k \rangle_k := \frac{1}{K} \sum_k F(\rho_{\mathrm{LHS}\_k}, \rho_{\mathrm{RHS}})$. 

As discussed in Appendix~\ref{app:Achievability of the 2-design approximation}, the quality of the 2-design approximation depends on the circuit depth, which determines the size $m_k$ of the averaging set $\{\superD_j^{(k)}\}$ and its randomness. We simulate increasing circuit depth by composing the $U_{\mathrm{TS}}$ circuit $R$ times since the TSAC is a multi-round protocol and the circuit applies multiple rounds on the computational qubits.

The results for $n=3$ qubits with an initial thermal state $\rho_0$ corresponding to a ground-state population $P=0.8$ are shown in Fig.~\ref{fig:gda_validation}. The $R=0$ case is a baseline, representing the fidelity between the un-averaged noise channel $\Lambda^{(k)}(\rho_0)$ (i.e., $\superD_j = \mathbb{I}$) and the GDA model. The fidelity at $R=0$ is low ($\approx 0.82$), indicating a deviation. However, applying the averaging effect provided by just one round of the TSAC circuit ($R=1$) causes a jump in fidelity to $\approx 0.925$. The fidelity converges to a value of $\approx 0.937$ for $R \ge 2$.

This result validates the 2-design assumption, Eq.~\eqref{eq:2-design}, for the circuit. While the fidelity in Fig.~\ref{fig:gda_validation} converges to $\langle F_k \rangle_k \approx 0.937$ rather than perfectly to 1, this high fidelity is sufficient to validate the GDA. The GDA framework itself is an approximation to first order in $p\,\nTG$, as shown in Eq.~\eqref{app:eq:first_order}, neglecting terms of $\mathcal{O}(p^2\nTG^2)$ and higher. The deviation of the 2-design approximation, Eq.~\eqref{eq:2-design}, which our numerical test quantifies as an infidelity $1-F \approx 0.063$, enters the final GDA formula only after being multiplied by the first-order prefactor $p\,\nTG$. Thus, this contributes only as a higher-order correction, $\mathcal{O}(p\,\nTG \cdot (1-F))$. This numerically confirms that the GDA does not require the circuit's prefix unitaries to form a perfect 2-design, but only one that is "good enough".

\begin{figure}[h!]
\includegraphics[width=8.6cm]{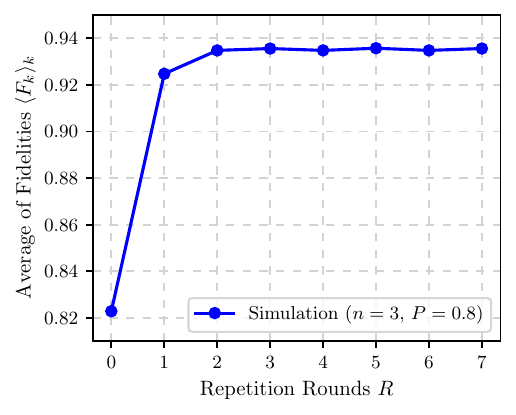} 
\caption{Numerical validation of the 2-design approximation, Eq.~\eqref{eq:2-design}, for the $n=3$ TSAC circuit. The plot shows the average fidelity $\langle F_k \rangle_k$ between the true averaged noise channel (LHS) and the GDA model (RHS) as a function of circuit depth, simulated by repetition rounds $R$. The initial state is a thermal state with $P=0.8$. The $R=0$ point represents the baseline fidelity without the averaging effect. The rapid convergence to high fidelity for $R \ge 1$ demonstrates the effectiveness of the 2-design assumption.}
\label{fig:gda_validation}
\end{figure}


\section{Derivation of the Analytical Cooling Limit (Proposition~\ref{prop:noisy_limit})}
\label{app:derivation_noisy_limit}

{\noindent}In this appendix, we provide the full derivation for the analytical cooling limit of noisy TSAC, as stated in Proposition~\ref{prop:noisy_limit}. We first construct the noisy transition matrix, then solve its eigenvalue problem to find the steady state, and finally show how this general solution recovers the ideal, noiseless limit.


\subsection{Constructing the Noisy Transition Matrix}
\label{app:derivation_noisy_matrix}

{\noindent}The dynamics of ideal TSAC acting on the vector of diagonal elements $\vec{v}$ (of the $n_{\mathrm{c}}$ computational qubits) can be described by the ideal $d_{\mathrm{c}} \times d_{\mathrm{c}}$ transition matrix~$T$, where $d_{\mathrm{c}} = 2^{n_{\mathrm{c}}}$~\cite{Raeisi2019, Farahmand2022}:
\begin{equation}\label{eq:transition_ideal_appendix}T=\frac{1}{z}\begin{pmatrix}e^\epsilon&e^\epsilon&0&\cdots&0\\e^{-\epsilon}&0&e^\epsilon&\cdots&0\\0&e^{-\epsilon}&0&\cdots&0\\ \vdots & \vdots & \vdots & \ddots & \vdots \\0&0&\cdots&e^{-\epsilon}&e^{-\epsilon}\end{pmatrix},
\end{equation}
where $z=e^{\epsilon}+e^{-\epsilon}$.

We focus on the dynamics of the $n_{\mathrm{c}}$ computational qubits. The GDA noise, calculated on the full $n$-qubit system, is here modeled as an effective channel $\mathcal{C}_\eta$ on the $d_{\mathrm{c}} = 2^{n_{\mathrm{c}}}$ dimensional computational subspace. The effect of this channel on the diagonal elements of the computational density matrix $\rho_C$ is given by
\begin{equation}
   \mathcal{C}_\eta(\rho_C) = (1-\eta)\rho_C + \eta \frac{I_{d_{\mathrm{c}}}}{d_{\mathrm{c}}},
\end{equation}
where~$\eta$ is the effective noise parameter. Since the TSAC dynamics can be described by the evolution of the vector of diagonal elements of the density matrix, $\vec{v}$, we first determine how $\mathcal{C}_\eta$ acts on this vector. 

Given a density matrix $\rho_C$ such that $\vec{v} = \text{diag}(\rho_C) =(v_1,v_2...,v_{d_{\mathrm{c}}})^T$, the effect of the depolarizing channel $\mathcal{C}_\eta$  can be regarded as a linear map:
\begin{equation}
\text{diag}(\rho_C)\xrightarrow{\mathcal{C}_\eta}(1-\eta) \text{diag}(\rho_C) + \frac{\eta}{d_{\mathrm{c}}}\sum_j^{d_{\mathrm{c}}} v_j(1,1...,1)^T .
\end{equation}
The matrix representation of this map is:
\begin{equation}
\mathbf{C}_\eta =
\begin{pmatrix}
1-\eta+\frac{\eta}{d_{\mathrm{c}}} & \frac{\eta}{d_{\mathrm{c}}} & \cdots & \frac{\eta}{d_{\mathrm{c}}} \\
\frac{\eta}{d_{\mathrm{c}}} & 1-\eta+\frac{\eta}{d_{\mathrm{c}}} & \cdots & \frac{\eta}{d_{\mathrm{c}}} \\
\vdots & \vdots & \ddots & \vdots \\
\frac{\eta}{d_{\mathrm{c}}} & \frac{\eta}{d_{\mathrm{c}}} & \cdots & 1-\eta+\frac{\eta}{d_{\mathrm{c}}}
\end{pmatrix}.
\end{equation}
The noisy TSAC iteration is a composition of the ideal TSAC channel (represented by matrix~$T$) and the noise channel (represented by matrix $\mathbf{C}_\eta$). This leads to the following lemma.

\begin{figure}[b]
    \centering 
    \begin{overpic}[width=8.6cm, trim=0 0 0 19, clip]{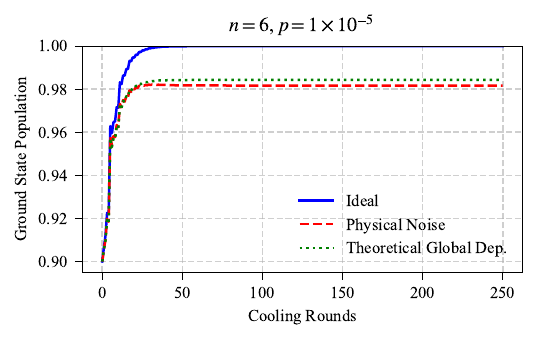}
        \put(15, 50){\textbf{(a)}}
    \end{overpic}
    \begin{overpic}[width=8.6cm, trim=0 0 0 17.5, clip]{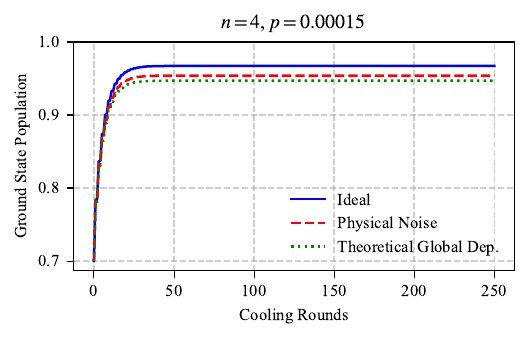}
        \put(15, 50){\textbf{(b)}}
    \end{overpic}
    
    \caption{Ground state population vs. cooling rounds.
    Simulations compare the ideal (noiseless) case, the gate-dependent timekeeping (physical) noise simulation, and the global depolarizing model (theoretical). 
    Results are shown for different numbers of qubits ($n$) and error probabilities ($p$): 
    (a) $n=6, p=1 \times 10^{-5}$. 
    (b) $n=4, p=1.5 \times 10^{-4}$.}
    \label{fig:cooling_rounds_all_cases}
\end{figure}

\begin{lemma}[Noisy Transition Matrix for TSAC]
\label{lem:noisy_matrix}
The transition matrix~$T^{\,\prime}_{\eta}$ for the TSAC protocol under the Global Depolarizing Approximation is given by $T^{\,\prime}_{\eta} = \mathbf{C}_\eta T$, where~$T$ is the ideal transition matrix from Eq.~\eqref{eq:transition_ideal_appendix}. This results in the $d_{\mathrm{c}} \times d_{\mathrm{c}}$ matrix:
\begin{equation}\label{app:Trans_matrix_noise}
        T^{\,\prime}_{\eta}=\left(\begin{array}{cccccc}
        a_{\eta,+}&a_{\eta,+}&\frac{\eta}{d_{\mathrm{c}}}&\cdots&\frac{\eta}{d_{\mathrm{c}}}\\
        a_{\eta,-}&\frac{\eta}{d_{\mathrm{c}}}&a_{\eta,+}&\cdots&\frac{\eta}{d_{\mathrm{c}}}\\
        \frac{\eta}{d_{\mathrm{c}}}&a_{\eta,-}&\frac{\eta}{d_{\mathrm{c}}}&\cdots&\frac{\eta}{d_{\mathrm{c}}}\\
        \vdots&\vdots&\ddots&\ddots&\vdots\\
        \frac{\eta}{d_{\mathrm{c}}}&\frac{\eta}{d_{\mathrm{c}}}&\cdots&a_{\eta,-}&a_{\eta,-}\end{array}\right),
\end{equation}
where $a_{\eta,\pm} := (1-\eta)\frac{e^{\pm\epsilon}}{z}+ \frac{\eta}{d_{\mathrm{c}}} $, $d_{\mathrm{c}} = 2^{n_{\mathrm{c}}}$ is the dimension of $\HH_C$. Here, $\eta$ is the parameter of the global depolarizing error.
\end{lemma}

\begin{proof}
The elements of~$T^{\,\prime}_{\eta}$ are $(T^{\,\prime}_{\eta})_{ij} = \sum_k (\mathbf{C}_\eta)_{ik} T_{kj}$. Due to the structure of $\mathbf{C}_\eta$, this becomes $(T^{\,\prime}_{\eta})_{ij} = (1-\eta)T_{ij} + \frac{\eta}{d_{\mathrm{c}}} \sum_k T_{kj}$. Since each column of the ideal matrix~$T$ sums to 1, we have $\sum_k T_{kj} = 1$. Thus, $(T^{\,\prime}_{\eta})_{ij} = (1-\eta)T_{ij} + \frac{\eta}{d_{\mathrm{c}}}$. Applying this to the specific structure of~$T$ in Eq.~\eqref{eq:transition_ideal_appendix} yields the matrix elements in Eq.~\eqref{app:Trans_matrix_noise}. 
\end{proof}

While we focus on the steady state below, the matrix $T'$ in Eq.~(\ref{app:Trans_matrix_noise}) already allows us to predict the cooling dynamics at finite rounds. The dynamics of the cooling process are presented in Fig.~\ref{fig:cooling_rounds_all_cases}, showing the ground-state population as a function of cooling rounds. The noise is modeled using the gate-dependent timekeeping error channel described in Sec.~\ref{app:specific_noise}, characterized by the CNOT error probability $p$. Performance is evaluated using the target qubit's ground-state population $P_n$. In the figures, the label "theoretical" refers to predictions from the GDA model. "Physical" refers to results from gate-level simulations of timekeeping noise. 


\subsection{Solving for the Steady State}

{\noindent}Finding the cooling limit is equivalent to finding the eigenvector $\vec{v}$ of~$T^{\,\prime}_{\eta}$ with eigenvalue 1, i.e., solving $T^{\,\prime}_{\eta}\vec{v} = \vec{v}$~\cite{Raeisi2019}.
Let $E := \tanh(\epsilon)$ and $C := \frac{2\eta}{d_{\mathrm{c}}(1-\eta)}$.
The noisy transition matrix~$T^{\,\prime}_{\eta}$ can be rewritten as
\begin{align}
        T^{\,\prime}_{\eta}=\frac{(1-\eta)}{2}\left(\begin{array}{cccccc}
        1+E + C&1+E + C&C&\cdots&C\\
        1-E + C&C&1+E + C&\cdots&C\\
        C&1-E + C&C&\cdots&C\\
        C&C&\cdots&\ddots&\vdots\\
        C&C&\cdots&1-E + C&1-E + C\end{array}\right).
\end{align}
The matrix equation $T^{\,\prime}_{\eta}\vec{v} = \vec{v}$ can be rewritten as
\begin{equation}\label{eq:matrix_equation}
    \mathbf{M}\vec{v} = \frac{2}{1-\eta}\vec{v} =(2+d_{\mathrm{c}} C)\vec{v},
\end{equation}
where $\mathbf{M}$ is a matrix whose elements are shifted with respect to the elements of~$T^{\,\prime}_{\eta}$, $\mathbf{M}:= \frac{2}{1-\eta}T^{\,\prime}_{\eta}$.
The matrix equation Eq.~(\ref{eq:matrix_equation}) corresponds to a system of linear equations for the components $v_k$ of the vector $\vec{v}$, 
\begin{subequations}
\begin{align}
    (1+ E)(v_1 + v_2) + C\sum_{k=1}^{d_{\mathrm{c}}} v_k &= (2+d_{\mathrm{c}} C)v_1, \label{app:Boundary1} \\
    (1- E)v_{k-1} + (1 + E)v_{k+1} + C\sum_{k=1}^{d_{\mathrm{c}}} v_k &= (2+d_{\mathrm{c}} C)v_k \,\,\,\, \text{for}\,\,\,\, 2\leq k\leq d_{\mathrm{c}}-1, \label{app:recurrenceOriginal} \\
    (1- E)(v_{d_{\mathrm{c}}-1} + v_{d_{\mathrm{c}}}) + C\sum_{k=1}^{d_{\mathrm{c}}} v_k &= (2+d_{\mathrm{c}} C)v_{d_{\mathrm{c}}}. \label{app:Boundary2} 
\end{align}
\end{subequations}
Since the vector $\vec{v}$ is a probability distribution, $\sum_k v_k = 1$. The recurrence relation in Eq.~\eqref{app:recurrenceOriginal} simplifies to
\begin{equation}\label{app:recurrence}
    (1- E)v_{k-1} + (1 + E)v_{k+1} + C = (2+d_{\mathrm{c}} C)v_k.
\end{equation}
This is a second-order linear inhomogeneous recurrence relation. We can find its general solution by introducing $v_k^{*} = v_k - 1/d_{\mathrm{c}}$, which transforms the equation into a homogeneous one,
\begin{equation}
    (1- E)v_{k-1}^{*} + (1 + E)v_{k+1}^{*} = (2+d_{\mathrm{c}} C)v_k^{*}.
\end{equation}
The characteristic equation of this recurrence has two roots.
The general solution to a second-order linear homogeneous recurrence relation is a linear combination of terms involving these roots. Therefore, the solution for the homogeneous part, $v_k^{*}$, is of the form $z_1 \lambda_1^{k-1} + z_2 \lambda_2^{k-1}$. Recalling that $v_k = v_k^{*} + 1/d_{\mathrm{c}}$, the general solution for the steady-state probability distribution $v_k$ is
\begin{equation}\label{app:general_form_vk}
    v_k = z_1 \lambda_1^{k-1} + z_2 \lambda_2^{k-1} + \frac{1}{d_{\mathrm{c}}}.
\end{equation}
This establishes the analytical form presented in Proposition~\ref{prop:noisy_limit} in the main text. The remaining task is to determine the explicit expressions for the parameters $\lambda_1, \lambda_2, z_1,$ and $z_2$.


\subsection{General Solution and Parameter Definitions}

{\noindent}The characteristic equation corresponding to the homogeneous recurrence is
\begin{equation}
    (1+E)\lambda^2 - (2+d_{\mathrm{c}} C)\lambda + (1-E) = 0.
\end{equation}
The two roots, which we denote $\lambda_1$ and $\lambda_2$, are given by
\begin{subequations}
\begin{align}
    \lambda_1 &= \frac{2+d_{\mathrm{c}} C + \sqrt{(2+d_{\mathrm{c}} C)^2 - 4(1- E^2)}}{2(1+E)},\\
    \lambda_2 &= \frac{2+d_{\mathrm{c}} C - \sqrt{(2+d_{\mathrm{c}} C)^2 - 4(1- E^2)}}{2(1+E)}.
\end{align}
\end{subequations}
From Vieta's formulas, the product of the roots is independent of the noise term $C$, providing a useful constraint:
\begin{equation}
\lambda_1 \lambda_2 = \frac{1-E}{1+E} = e^{-2\epsilon}.
\end{equation}
The coefficients $z_1$ and $z_2$ are determined by substituting the general solution, Eq.~\eqref{app:general_form_vk}, back into the two boundary condition equations,~\eqref{app:Boundary1} and~\eqref{app:Boundary2}. This yields the explicit expressions for the coefficients, 
\begin{subequations}
\begin{align}
    z_1 &= \frac{E \left(\lambda_2^{d_{\mathrm{c}}}-1\right)}{d_{\mathrm{c}} (\lambda_2-1) (E+1) \left(\lambda_2^{d_{\mathrm{c}}}-\lambda_1^{d_{\mathrm{c}}}\right)}, \label{eq:z1_final} \\
    z_2 &= \frac{E \left(\lambda_1^{d_{\mathrm{c}}}-1\right)}{d_{\mathrm{c}} (\lambda_1-1) (E+1) \left(\lambda_1^{d_{\mathrm{c}}}-\lambda_2^{d_{\mathrm{c}}}\right)}. \label{eq:z2_final}
\end{align}
\end{subequations}
This completes the explicit determination of all parameters in the steady-state solution.


\subsection{Reduction to the Ideal Case}

{\noindent}We now verify that our general solution correctly reduces to the known noiseless limit as $\eta \to 0$. In this limit, $C \to 0$, and the eigenvalues become:
\begin{equation}
    \lambda_1^{(0)} = 1, \qquad \lambda_2^{(0)} = \frac{1-E}{1+E} = e^{-2\epsilon}.
\end{equation}
We can now find the coefficients in this limit. For $z_1$, substituting $\lambda_1 = 1$ into its expression leads to a simplified form (after canceling the $(\lambda_2^{d_{\mathrm{c}}}-1)$ terms), 
\begin{align}
z_1^{(0)} = \frac{E}{d_{\mathrm{c}}(\lambda_2^{(0)}-1)(E+1)}
= \frac{E}{d_{\mathrm{c}}\left(\frac{1-E}{1+E}-1\right)(E+1)}
= -\frac{1}{d_{\mathrm{c}}}.
\end{align}
For $z_2$, we use the normalization constraint $\sum_{k=1}^{d_{\mathrm{c}}} (v_k' - 1/d_{\mathrm{c}}) = 0$. In the ideal limit, this becomes
\begin{align}
    \sum_{k=1}^{d_{\mathrm{c}}} \left(z_1^{(0)} (\lambda_1^{(0)})^{k-1} + z_2^{(0)} (\lambda_2^{(0)})^{k-1}\right) = 
    d_{\mathrm{c}} z_1^{(0)} + z_2^{(0)} \frac{1 - (\lambda_2^{(0)})^{d_{\mathrm{c}}}}{1-\lambda_2^{(0)}} = 0.
\end{align}
Substituting $z_1^{(0)} = -1/d_{\mathrm{c}}$:
\begin{equation}
    -1 + z_2^{(0)} \frac{1 - (e^{-2\epsilon})^{d_{\mathrm{c}}}}{1-e^{-2\epsilon}} = 0 \implies z_2^{(0)} = \frac{1-e^{-2\epsilon}}{1-e^{-2d_{\mathrm{c}}\epsilon}}.
\end{equation}
The final distribution is:
\begin{equation}
    v_k^{(\text{ideal})} = -\frac{1}{d_{\mathrm{c}}} (1)^{k-1} + z_2^{(0)} (e^{-2\epsilon})^{k-1} + \frac{1}{d_{\mathrm{c}}} = z_2^{(0)} (e^{-2\epsilon})^{k-1}.
\end{equation}
This is the ideal cooling limit given in~\cite{Raeisi2019}.


\subsection{Perturbative Analysis and Physical Interpretation}

{\noindent}To understand the impact of small noise, we perform a first-order perturbative expansion of the eigenvalues for $\eta \ll 1$. Let $\lambda_1 = 1 + \delta$, where~$\delta$ is a small correction. Substituting this into the characteristic equation and keeping terms linear in~$\delta$ and~$\eta$ (approximating $d_{\mathrm{c}} C \approx 2\eta$), we find
\begin{subequations}
\begin{align}
    (1+E)(1+2\delta) - (2+2\eta)(1+\delta) + (1-E) &\approx 0, \\[1mm]
    2E\delta - 2\eta &\approx 0, 
\end{align}
\end{subequations}
which implies $\delta \approx \tfrac{\eta}{E}$. 
The eigenvalues to first order are therefore
\begin{equation}
    \lambda_1 \approx 1 + \frac{\eta}{\tanh(\epsilon)} > 1, \qquad
    \lambda_2 \approx e^{-2\epsilon}\left(1 - \frac{\eta}{\tanh(\epsilon)}\right).
\end{equation}
Noise introduces an eigenvalue $\lambda_1 > 1$. By continuity from the ideal case, the corresponding coefficients satisfy $z_1 < 0$ and $z_2 > 0$.

The deviation from the uniform background, $\Delta v_k = v_k - 1/d_{\mathrm{c}}$, must sum to zero:
\begin{equation}
    \sum_{k=1}^{d_{\mathrm{c}}} (z_1 \lambda_1^{k-1} + z_2 \lambda_2^{k-1}) = 0.
\end{equation}
Since $\lambda_1 > 1$, the sum $\sum \lambda_1^{k-1} = \frac{\lambda_1^{d_{\mathrm{c}}} - 1}{\lambda_1 - 1}$ grows exponentially with $d_{\mathrm{c}}$. To satisfy the normalization and positivity ($v_k \ge 0$) constraints simultaneously, both coefficients $|z_1|$ and $z_2$ must be suppressed as $d_{\mathrm{c}}$ increases.

This suppression forces  $\Delta v_k$, to vanish in the large-$d_{\mathrm{c}}$ limit, causing the cooling limit to converge to a high-temperature uniform distribution, $v_k \to 1/d_{\mathrm{c}}$. This contrasts with the ideal case, where increasing the number of computational qubits $n_{\mathrm{c}}$ always improves the cooling performance. Therefore, noise introduces a trade-off, leading to the existence of an optimal number of qubits~$n$ (and thus $n_{\mathrm{c}}$) that maximizes the cooling effect before the noise-induced suppression dominates. 

To validate this analytical prediction and visualize the specific impact of noise accumulation, we compare the GDA model against gate-level physical simulations.
Fig.~\ref{fig:cooling_limit} provides the detailed scaling behavior underlying the summary in Fig. 2 of the main text. The system is initialized with an initial ground-state population of $P_{initial} \approx 0.85$. The figure plots the final ground-state population $P_n$ as a function of $n$ for various error probabilities ($p \in \{10^{-4}, 10^{-6}, 10^{-8}\}$). These results explicitly identify a transition point in the cooling performance: unlike the ideal case, as $n$ increases beyond an optimal point, the accumulated noise causes the performance to degrade.

\begin{figure}
    \centering
    \includegraphics[width=8.6cm]{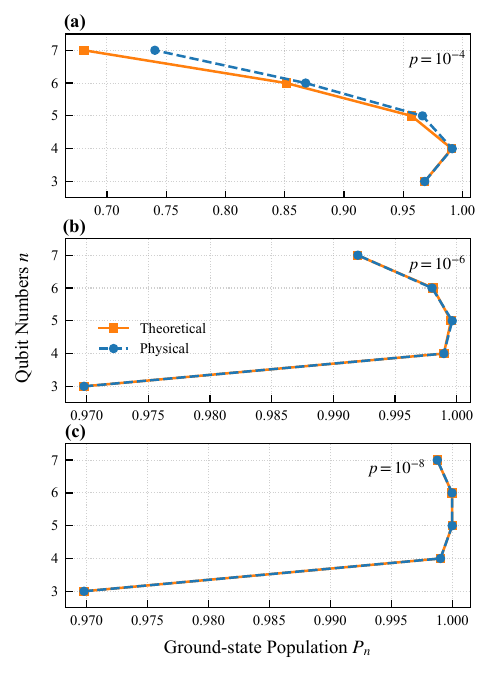} 
    \caption{\label{fig:cooling_limit}
    Cooling performance of the TSAC protocol under noise for different numbers of qubits. 
    The final ground-state population of the target qubit $P_n$ (defined in Sec. III B) is shown as a function of the total number of qubits ($n$) for three fixed CNOT error probabilities: \textbf{(a)} $p = 10^{-4}$, \textbf{(b)} $p=10^{-6}$, and \textbf{(c)} $p=10^{-8}$. 
    All panels compare results from physical noise simulations (circles, dashed lines) and the theoretical GDA model (squares, solid lines).}
\end{figure}

\section{Analysis of Specific Noise Models}
\label{app:specific_noise}

{\noindent}Suppose that a noisy circuit that fulfills our assumptions with $\nTG$ number of two-qubit gates is given. To demonstrate the versatility of the GDA framework presented in Theorem~\ref{thm:gda}, we explicitly calculate the effective depolarizing strength $\eta = p \nTG (1-q)$ for two common physical noise models. The key is to correctly define the $n$-qubit error channel $\mathcal{E}_g$ from its local 2-qubit counterpart and then compute the corresponding parameters $p$ and $q$.

We define the $n$-qubit error channel $\mathcal{E}_g$ as the natural extension of a 2-qubit channel $\mathcal{B}$ acting on a designated pair of qubits. Specifically, $\mathcal{E}_g$ acts as $\mathcal{B}$ on the target pair and as the identity channel $\mathbb{I}$ on the remaining $n-2$ spectator qubits.

\begin{lemma}[GDA Parameter for Two-Qubit Bit-Flip Noise]
\label{lem:bitflip_noise}
Consider a two-qubit bit-flip channel $\mathcal{B}$ with physical error parameter $\chi$, characterized by the Kraus operators
\begin{subequations}
    \begin{align}
        E_{II} &= (1-\chi) I \otimes I, \\[1mm]
        E_{XI} &= \sqrt{\chi(1-\chi)} X \otimes I, \\[1mm]
        E_{IX} &= \sqrt{\chi(1-\chi)} I \otimes X, \\[1mm]
        E_{XX} &= \chi X \otimes X.
    \end{align}
\end{subequations}
When its $n$-qubit extension acts as the error source $\mathcal{E}_g$ for two-qubit gates, the corresponding effective global depolarizing strength is:
\begin{equation}
    \eta_B = (2\chi-\chi^2) \cdot \nTG \cdot \frac{d^2}{d^2-1}.
\end{equation}
\end{lemma}

\begin{proof}
The 2-qubit channel $\mathcal{B}$ can be written in the form of $\mathcal{E}_g(\rho_{2q}) = (1-p_{2q})\rho_{2q} + p_{2q}\Lambda_{2q}(\rho_{2q})$, where the local error probability is $p_{2q} = 1 - (1-\chi)^2 = 2\chi - \chi^2$. 

When extended to the $n$-qubit system, the error probability $p$ remains the same, as an error either occurs on the 2-qubit subsystem or it does not: $p = p_{2q} = 2\chi - \chi^2$.
The $n$-qubit normalized error channel is $\Lambda_g = \Lambda_{2q} \otimes \mathbb{I}_{n-2}$, where $\Lambda_{2q}$ is the normalized mixture of the non-identity Pauli error channels corresponding to the Kraus operators $E_{XI}, E_{IX}, E_{XX}$. Since $\Lambda_g$ is a traceless Pauli channel, i.e., $\Tr(\Lambda_g)=0$, according to Theorem~\ref{thm:gda}, we have the effective global depolarizing strength
\begin{align}
    \eta_B &= p \cdot \nTG \cdot (1-q) 
    \,=\, (2\chi-\chi^2) \cdot \nTG \cdot \left(1 - \frac{-1}{d^2-1}\right) 
    \,=\, (2\chi-\chi^2) \cdot \nTG \cdot \frac{d^2}{d^2-1}.
\end{align}
This completes the proof.
\end{proof}

\begin{lemma}[GDA Parameter for Timekeeping Noise]
\label{lem:timekeeping_noise}
Consider the 2-qubit timekeeping noise channel associated with a CNOT gate, described by a physical parameter $p$ and the Kraus operators~\cite{Xureb2023}
\begin{subequations}
 \begin{align}
         A_1 &=\sqrt{1-p} \ I\otimes I,\\[1mm]
         A_2 &=\sqrt{p}\,(|1\rangle\!\langle1|\otimes X +|0\rangle\!\langle0|\otimes I).
 \end{align}
\end{subequations}
When its $n$-qubit extension acts as the error source $\mathcal{E}_g$, the corresponding effective global depolarizing strength is
\begin{equation}
    \eta_{\hspace*{0.75pt}\scriptscriptstyle\mathrm{T}} = p \cdot \nTG \cdot \frac{3d^2/4}{d^2 - 1}.
\end{equation}
\end{lemma}

\begin{proof}
For the Kraus operators of the timekeeping noise, the local error probability is $ p$, and consequently, the $n$-qubit error probability is $ p$.
The normalized $n$-qubit error channel is $\Lambda_g(\rho) = \frac{1}{p} (A_2 \otimes I_{n-2}) \rho (A_2^\dagger \otimes I_{n-2})$. This channel has a single Kraus operator, $K = (1/\sqrt{p}) A_2 \otimes I_{n-2} = K_{2q} \otimes I_{n-2}$, where $K_{2q} := |1\rangle\!\langle1|\otimes X +|0\rangle\!\langle0|\otimes I$.

To calculate the effective depolarizing probability, we compute the trace of the superoperator $\Lambda_g$. For a channel with a single Kraus operator $K$, this trace is given by $\Tr(\Lambda_g) = |\Tr(K)|^2$. The trace of the $n$-qubit operator $K$ reads
\begin{align*}
    \Tr(K) &= \Tr(K_{2q} \otimes I_{n-2}) = \frac{d}{2}.
\end{align*}
Therefore, the trace of the superoperator is $\Tr(\Lambda_g) = (d/2)^2 = d^2/4$. Now, we compute the parameter $q$ using its definition,
\begin{equation}
    q = \frac{\Tr(\Lambda_g) - 1}{d^2 - 1} = \frac{d^2/4 - 1}{d^2 - 1}.
\end{equation}
Finally, we compute the factor $(1-q)$, 
\begin{equation}
    1 - q = 1 - \frac{d^2/4 - 1}{d^2 - 1}= \frac{3d^2/4}{d^2 - 1}.
\end{equation}
Substituting $p_g = \alpha$ and this value of $(1-q)$ into the definition of~$\eta$ gives the final result, 
\begin{equation}
    \eta_{\hspace*{0.75pt}\scriptscriptstyle\mathrm{T}} = p \cdot \nTG \cdot (1-q) = p \cdot \nTG \cdot \frac{3d^2/4}{d^2 - 1}.
\end{equation}
This completes the proof.
\end{proof}


\section{Numerical Implementation Details}
\label{app:supplementary_results}


\subsection{Numerical Implementation of TSAC}
\label{app:simulation_details}

{\noindent}In our numerical simulations, we used Qiskit~\cite{javadiabhari2024}. The TSAC compression unitary $U_{\mathrm{TS}}$ was constructed for a total of~$n$ qubits (comprising $n_{\mathrm{c}} = n-1$ computational/target qubits and one reset qubit). To generalize to arbitrary~$n$, we employed the recursive ladder structure detailed in Algorithm~\ref{alg:build_U}. The circuit uses a sequence of multi-controlled X (MCX) gates. Crucially, since decomposing an MCX gate with $k$ controls requires $\mathcal{O}(k)$ or more physical CNOT gates, the total gate count $\nTG$ grows rapidly with~$n$, directly driving the increase in the effective noise parameter~$\eta$.
     
\begin{algorithm}[H]
\caption{Construction of the TSAC compression circuit $U_{\mathrm{TS}}$ for~$n$ qubits}
\label{alg:build_U}
\begin{algorithmic}[1]
\Require Number of qubits~$n$ (indices $0, 1, \dots, n-1$).
\Ensure Quantum Circuit $U_{\mathrm{TS}}$.
\State Initialize a quantum circuit with~$n$ qubits.
\State Apply $X$ gate on the reset qubit $q_{n-1}$.

\Statex \textit{// Step 1: Down-ladder of MCX gates}
\For{$i = n-2$ \textbf{down to} $0$}
    \State Define controls $C = \{q_{i+1}, \dots, q_{n-1}\}$.
    \State Apply Multi-Controlled $X$ gate: $\text{MCX}(C \to q_i)$.
\EndFor

\Statex \textit{// Step 2: Top MCX gate}
\If{$n > 1$}
    \State Define controls $C_{\text{all}} = \{q_0, \dots, q_{n-2}\}$.
    \State Apply $\text{MCX}(C_{\text{all}} \to q_{n-1})$.
\EndIf

\State Apply $X$ gate on the reset qubit $q_{n-1}$.

\Statex \textit{// Step 3: Up-ladder of MCX gates }
\For{$i = 0$ \textbf{to} $n-2$}
    \State Define controls $C = \{q_{i+1}, \dots, q_{n-1}\}$.
    \State Apply $\text{MCX}(C \to q_i)$.
\EndFor

\State Apply $X$ gate on the reset qubit $q_{n-1}$.
\State \textbf{return} Circuit
\end{algorithmic}
\end{algorithm}

Figure~\ref{fig:tsac_circuit_n3_app} shows the circuit for $n=3$ (i.e., $n_{\mathrm{c}}=2$), which is the smallest system analyzed in Fig. 2 of the main text.
    
\begin{figure}[h!]
    \centering
    \includegraphics[scale=1.4]{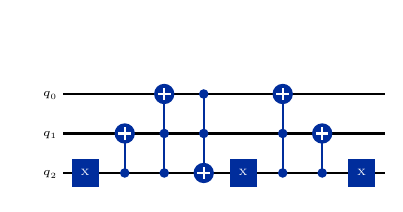}
    \caption{The $n=3$ TSAC compression circuit $U_{\mathrm{TS}}$ used in the numerical simulation~\cite{Shende2024}. As indicated by the labels, $q_0$ is the Target qubit, $q_2$ is the Refresh qubit, and $q_1$ is a computational qubit.}
    \label{fig:tsac_circuit_n3_app}
\end{figure}

To calculate the gate counts and run the noise simulations, this $U_{\mathrm{TS}}$ circuit (and all other $n$-qubit circuits) was transpiled using \texttt{qiskit.transpile} to the basis gate set $\{\texttt{cx}, \texttt{sx}, \texttt{rz}\}$.

The GDA effective noise strength~$\eta$ is given by Lemma~\ref{lem:timekeeping_noise}:
\begin{equation}
    \eta_{\hspace*{0.75pt}\scriptscriptstyle\mathrm{T}} = p \cdot \nTG \cdot \frac{3d^2/4}{d^2 - 1}
    \nonumber
\end{equation}
We illustrate each term using the $n=3$ case as an example:
\begin{itemize} 
    \item{$p$ is the physical CNOT error probability.}
    \item{$d$ is the total Hilbert-space dimension. For $n=3$ total qubits, $d  = 2^3$.}
    \item{$\nTG$ is the total number of CNOT (\texttt{cx}) gates after transpiling $U_{\mathrm{TS}}$ to the basis set. For the $n=3$ circuit (Fig.~\ref{fig:tsac_circuit_n3_app}), we count $\nTG = 20$.}
    \item{The dimension-dependent factor is $\frac{3d^2/4}{d^2-1} =  \approx 0.7619$ with $d=8$.}
\end{itemize}
Thus, for $n=3$ and a given error
probability, e.g., $p = 10^{-3}$, the effective GDA strength $\eta_{\hspace*{0.75pt}\scriptscriptstyle\mathrm{T}}$ is:
\begin{equation}
    \eta_{\hspace*{0.75pt}\scriptscriptstyle\mathrm{T}} = 10^{-3} \cdot 20 \cdot 0.7619 \approx 1.52 \times 10^{-2}
    \nonumber
\end{equation}

\subsection{Details of the DC Mirror Protocol}
{\noindent}We specifically consider the \textit{Mirror protocol} for the unitary $U_{\text{DC}}$, as proposed in~\cite{Oftelie2024}. The construction of this unitary is based on the permutation of computational basis states. 

Let $x = x_1 x_2 \dots x_n$ be a binary string of length $n$, representing a basis state $|x\rangle$ in the computational basis. We define its mirror string $\bar{x}$ as the bitwise negation of $x$. For example, for a 3-qubit system, the mirror of $001$ is $110$.

The unitary $U_{\text{DC}}$ can be mathematically decomposed as a product of independent SWAP operations between these mirror pairs. Formally, we write the unitary as:
\begin{equation}\label{eq:udc}
    U_{\text{DC}} = \prod_{x \in \mathcal{S}} \text{SWAP}(|x\rangle, |\bar{x}\rangle),
\end{equation}
where $\text{SWAP}(|x\rangle, |\bar{x}\rangle)$ is the operator that exchanges the amplitudes of the states $|x\rangle$ and $|\bar{x}\rangle$, and acts as the identity on all other states.

The set $\mathcal{S}$ contains the pairs of states that require swapping. The rule to construct $\mathcal{S}$ is determined by the Hamming weight of the string and the state of the target qubit. Let $w(x)$ denote the Hamming weight (number of $1$s) of the string $x$. In a thermal state, a basis state with a lower Hamming weight has lower energy and, consequently, a higher population. To cool the target qubit (indexed as $1$), we want to map these high-population states to states where the target qubit is in the ground state $|0\rangle$.

Therefore, for every pair $\{x, \bar{x}\}$, we compare their weights. Assume $w(x) < w(\bar{x})$. We apply the swap rule:
\begin{itemize}
    \item If the target qubit of the lower-weight state $|x\rangle$ is in the excited state $|1\rangle$ (i.e., $x_1 = 1$), we include this pair in $\mathcal{S}$ and perform the swap.
    \item Otherwise, if the target qubit is already in $|0\rangle$, no operation is applied (identity).
\end{itemize}
This construction ensures that the states with higher initial populations are mapped to the subspace where the target qubit is in the ground state.
We now apply the GDA framework to analyze the performance of the DC protocol in the presence of noise.
When the unitary $U_{\text{DC}} = \prod_{x \in \mathcal{S}} \text{SWAP}(|x\rangle, |\bar{x}\rangle)$ as shown in Eq.~\eqref{eq:udc} is decomposed into elementary gates and implemented as a quantum circuit $\hat{\mathcal{U}}_{\,\text{DC}}$, the gate errors inevitably introduce noise. We denote the noisy circuit superoperator as $\hat{\mathcal{U}}^{\,\prime}_{\,\text{DC}}$. The final state of the target qubit after the noisy process is given by tracing out the auxiliary qubits:
\begin{equation}
    \rho^{\,\prime}_{\text{out}} = \Tr_{2,\ldots,n} \left( \hat{\mathcal{U}}^{\,\prime}_{\,\text{DC}} \left(\rho_{\text{th}}^{\otimes n} \right) \right).
\end{equation}

\begin{figure}[b]
    \centering
    \includegraphics[scale=1.8]{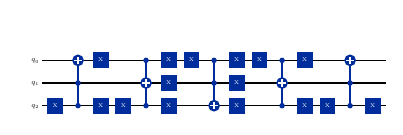}
    \caption{The $n=3$ DC circuit $U_{\mathrm{DC}}$ used in the numerical simulation. As indicated by the labels, $q_0$ is the target qubit.}
    \label{fig:DC_circuit_n3_app}
\end{figure}
Under the assumptions of the GDA, the noisy circuit can be approximated as the ideal unitary followed by a global depolarizing channel $\mathcal{C}_\eta$:
\begin{equation}
    \hat{\mathcal{U}}^{\,\prime}_{\,\text{DC}} \approx \mathcal{C}_\eta \circ \hat{\mathcal{U}}_{\,\text{DC}},
\end{equation}
where $\eta$ is the effective noise parameter derived from the circuit structure. 
Substituting this into the expression for the output state, we have:
\begin{equation}
    \rho^{\,\prime}_{\text{out}} \approx \Tr_{2,\ldots,n} \left( \mathcal{C}_\eta \left( \hat{\mathcal{U}}_{\,\text{DC}}\left(\rho_{\text{th}}^{\otimes n} \right) \right) \right).
\end{equation}
Recall that the action of the global depolarizing channel on an $n$-qubit state $\rho$ is $\mathcal{C}_\eta(\rho) = (1-\eta)\rho + \eta \frac{\mathbb{I}}{2^n}$.
Since the partial trace is a linear operation, the noisy output state of the target qubit is simply a depolarized version of the ideal output:
\begin{equation}
    \rho^{\,\prime}_{\text{out}} = (1-\eta)\rho_{\text{out}} + \eta \frac{\mathbb{I}}{2}.
\end{equation}
The quantum circuit for $\hat{\mathcal{U}}_{\,\text{DC}}$ used in the simulation is generated by the open-source Python library established in Ref.~\cite{Giuliano2025}. Figure~\ref{fig:DC_circuit_n3_app} shows the circuit for the $n=3$ case. This circuit is transpiled using \texttt{qiskit.transpile} to the basis gate set $\{\texttt{cx}, \texttt{sx}, \texttt{rz}\}$.

\end{document}